\documentclass[a4paper,onecolumn,superscriptaddress,11pt,unpublished]{quantumarticle}
\pdfoutput=1
\usepackage[utf8]{inputenc}
\usepackage[english]{babel}
\usepackage[T1]{fontenc}
\usepackage{amsmath}
\usepackage{hyperref}
\usepackage{microtype}
\usepackage[numbers,sort&compress]{natbib}

\usepackage{censor}

\usepackage{tikzit}
\input{zx.tikzdefs}

\tikzstyle{box}=[shape=rectangle, text height=1.5ex, text depth=0.25ex, yshift=0.5mm, fill=white, draw=black, minimum height=5mm, yshift=-0.5mm, minimum width=5mm, font={\small}]
\tikzstyle{Z dot}=[inner sep=0mm, minimum size=2mm, shape=circle, draw=black, fill={rgb,255: red,221; green,255; blue,221}]
\tikzstyle{Z phase dot}=[minimum size=5mm, font={\footnotesize\boldmath}, shape=rectangle, rounded corners=2mm, inner sep=0.2mm, outer sep=-2mm, scale=0.8, tikzit shape=circle, draw=black, fill={rgb,255: red,221; green,255; blue,221}, tikzit draw=blue]
\tikzstyle{X dot}=[Z dot, shape=circle, draw=black, fill={rgb,255: red,255; green,136; blue,136}]
\tikzstyle{X phase dot}=[Z phase dot, tikzit shape=circle, tikzit draw=blue, fill={rgb,255: red,255; green,136; blue,136}, font={\footnotesize\boldmath}]
\tikzstyle{hadamard}=[fill=yellow, draw=black, shape=rectangle, inner sep=0.6mm, minimum height=1.5mm, minimum width=1.5mm]
\tikzstyle{vertex}=[inner sep=0mm, minimum size=1mm, shape=circle, draw=black, fill=black]
\tikzstyle{vertex set}=[inner sep=0mm, minimum size=1mm, shape=circle, draw=black, fill=white, font={\footnotesize\boldmath}]

\tikzstyle{hadamard edge}=[-, dashed, dash pattern=on 2pt off 0.5pt, thick, draw={rgb,255: red,68; green,136; blue,255}]
\tikzstyle{brace edge}=[-, tikzit draw=blue, decorate, decoration={brace,amplitude=1mm,raise=-1mm}]
\tikzstyle{diredge}=[->]

\usepackage{xspace}
\usepackage{algorithm}
\usepackage{algorithmicx}
\usepackage[noend]{algpseudocode}

\newcommand{\CZ}{\ensuremath{\textrm{CZ}}\xspace}

\usepackage{bm}

\usepackage{amsmath,amsthm,amssymb}
\theoremstyle{definition}
\newtheorem{theorem}{Theorem}[section]

\newtheorem*{lemma*}{Lemma}
\newtheorem*{proposition*}{Proposition}
\newtheorem{proposition}[theorem]{Proposition}

\StopCensoring

\newcommand\new[1]{#1}

\title{Constructing quantum circuits with global gates}

\author{\censor{John van de Wetering}}
\affiliation{\censor{Institute for Computing and Information Sciences, Radboud University Nijmegen}}
\affiliation{\censor{University of Oxford Computer Science Department}}

\begin{document}

\maketitle

\begin{abstract}
	There are various gate sets that can be used to describe a quantum computation. A particularly popular gate set in the literature on quantum computing consists of arbitrary single-qubit gates and 2-qubit CNOT gates. A CNOT gate is however not always the natural multi-qubit interaction that can be implemented on a given physical quantum computer, necessitating a compilation step that transforms these CNOT gates to the native gate set.
	An especially interesting case where compilation is necessary is for ion trap quantum computers, where the natural entangling operation can act on more than 2 qubits and can even act globally on all qubits at once. This calls for an entirely different approach to constructing efficient circuits. In this paper we study the problem of converting a given circuit that uses 2-qubit gates to one that uses global gates. 
	Our three main contributions are as follows. 
	First, we find an efficient algorithm for transforming an arbitrary circuit consisting of Clifford gates and arbitrary phase gates into a circuit consisting of single-qubit gates and a number of global interactions proportional to the number of non-Clifford phases present in the original circuit. 
	Second, we find a general strategy to transform a global gate that targets all qubits into one that targets only a subset of the qubits. This approach scales linearly with the number of qubits that are not targeted, in contrast to the exponential scaling reported in (Maslov \& Nam, N.~J.~Phys.~2018).
	Third, we improve on the number of global gates required to synthesise an arbitrary $n$-qubit Clifford circuit from the $12n-18$ reported in (Maslov \& Nam, N.~J.~Phys.~2018) to $6n-8$.
\end{abstract}

\section{Introduction}

There are many different physical implementations of qubits for purposes of quantum computation, including superconducting electronic circuits~\cite{arute2019quantum}, photons~\cite{chapman2016experimental}, or nitrogen-vacancy centers~\cite{zargaleh2018nitrogen}.
A particularly promising approach is given by ion trap quantum computing~\cite{bruzewicz2019trapped}.
For most quantum computing architectures, the natural entangling operation targets just two qubits. Ion trap qubits form an exception in that the entangling operation, the \emph{M\o{}lmer-S\o{}rensen} interaction~\cite{molmersorensen1999}, naturally targets the entire array of qubits and can be focused to target 
\new{some subsets of qubits}~\cite{martinez2016compiling,debnath2016demonstration,grzesiak2020efficient,figgatt2019parallel}. 

In this setting, the natural set of quantum gates consists of arbitrary single-qubit unitaries, together with the multi-qubit M\o{}lmer-S\o{}rensen (MS) gate (also variously known as the XX coupling or the Ising gate). We distinguish here two different classes of MS gate: targeted and untargeted. The targeted gate can be applied to any specific subsets of qubits, while untargeted gates always target the entire array of available qubits. Of course, targeted gates are more versatile, but they are also in general harder to engineer. 
Following~\cite{maslov2018use} we will refer to (un)targeted M\o{}lmer-S\o{}rensen gates acting on any number of qubits as \emph{global} M\o{}lmer-S\o{}rensen gates, or GMS gate, for short. 

Single-qubit gates generally can be implemented with high fidelity, while MS gates are considerably more noisy. In addition, MS gates simply take more time to run on the machine than a single-qubit interaction. Hence, by minimising the number of MS gates needed to implement a given computation we can drastically decrease the cost of computation and increase its fidelity. 

Compared to the extensive literature on synthesising quantum circuits with two-qubit gates (in particular CNOT gates), the literature on constructing optimised quantum circuits with global gates is quite sparse. Nevertheless, some results are known.
The work in this paper was motivated by~\cite{maslov2018use}. They find general strategies for synthesising Clifford circuits, $n$-controlled Toffoli gates and the Quantum Fourier Transform, building on the results of~\cite{ivanov2015efficient}. Some of our results are direct generalisations of those found in~\cite{maslov2018use}. In~\cite{groenland2020signal} they find a way to construct arbitrary N-controlled phase rotations using just $2N$ untargeted GMS gates. 

All these result are dedicated to synthesising specific families of unitaries. This paper will focus instead on synthesising circuits with GMS gates from arbitrary circuits built using regular CNOT gates. Such a general approach was also taken in~\cite{martinez2016compiling} where they made a brute-force compiler for transforming unitaries into optimal GMS circuits. However, the cost of this algorithm rises steeply with the number of qubits, and they only used it to synthesise circuits for up to 4 qubits.

Our main result is an efficient algorithm for transforming a given quantum circuit consisting of Clifford gates and arbitrary single-qubit phase gates into a circuit consisting of single-qubit phase gates and GMS gates. We give a version using both targeted and untargeted GMS gates. Note that there is a naive algorithm that does this task (in the targeted case), that simply translates each CNOT gate into a 2-qubit GMS gate. This results in an unnecessarily large number of GMS gates for most circuits that contain a realistic number of CNOT gates. 
Our algorithm instead scales with the number of non-Clifford phase gates present in the original circuit. Specifically, given a $n$-qubit circuit containing $N$ non-Clifford phase gates, we require at most $N+6n-8$ targeted GMS gates, or $2N+O(n^2/\log n)$ untargeted GMS gates. For most useful circuits $N\gg n$ and hence the number of targeted GMS gates is proportional to $N$. 

Note that for many concrete small circuits such as $n$-controlled Toffoli gates, the handcrafted results of for instance~\cite{maslov2018use,groenland2020signal} require a lower number of GMS gates than the bounds we find. The value of our algorithm lies in its general applicability. 
Furthermore, for more complicated circuits that consist for instance of a series of Toffoli gates, a large number of the non-Clifford gates that are needed in a naive synthesis can be canceled out by any of the large number of the T-count optimising algorithms in the literature~\cite{amy2016t,heyfron2018efficient,kissinger2019tcount,deBeaudrap2020Techniques}, which might make our algorithm more efficient than simply synthesising each Toffoli gate in turn by the methods of Refs.~\cite{maslov2018use,groenland2020signal}.
\new{
We also expect our algorithm to work better when the input circuit naturally consists of gates resembling exponentiated Pauli's, such as for Trotterised chemistry circuits~\cite{phaseGadgetSynth,cowtan_generic_2020}.
}

Our construction works by first converting the original circuit into a series of exponentiated Pauli gates as in~\cite{litinski2019game,zhang2019optimizing}, and then each exponentiated Pauli is synthesised in turn. When this is done, there is still a remaining Clifford circuit that needs to be synthesised. Here we improve the $12n-18$ GMS gates needed to synthesise a Clifford circuit given in~\cite{maslov2018use} to $6n-8$, by using a more optimal Clifford normal form.

In~\cite{maslov2018use}, a construction is given for transforming a pair of untargeted GMS gates into a targeted GMS gate that does not interact with some chosen qubit. By iterating this procedure, targeted GMS gates can be constructed that target any subset of qubits. However, their description requires an exponential number of untargeted GMS gates as the number of targeted qubits decreases. We provide a variation on this construction that instead scales linearly as the number of targeted qubits decreases. This allows us to synthesise an arbitrary $n$-qubit Clifford circuit using $O(n^2/\log n)$ untargeted GMS gates.

The structure of the paper is as follows. In Section~\ref{sec:preliminaries} we cover the necessary preliminaries on exponentiated Pauli gates and GMS gates. Then in Section~\ref{sec:clifford} we cover the construction of arbitrary Clifford circuits using targeted GMS gates. In Section~\ref{sec:constructingCircuit} we present our main algorithm for transforming arbitrary circuits into GMS circuits.
Then in Section~\ref{sec:untargeted} we will find modifications to these results that work for untargeted GMS gates. Finally, in Section~\ref{sec:conclusion} we end with some concluding remarks.

\section{Preliminaries}\label{sec:preliminaries}
\subsection{Exponentiated Pauli gates}\label{sec:expPauli}

We denote by $X$, $Y$ and $Z$ the standard Pauli gates, and let $I$ denote the identity. We write $P_i$ for a Pauli $P\in\{X,Y,Z,I\}$ to denote that $P$ is applied to the qubit $i$. In this same manner, the gate $P_iP'_j$ denotes that $P$ is applied to $i$ while $P'$ is applied to $j$. Let $\vec{P}$ denote a `string' of Pauli gates , i.e.~an operator that applies a Pauli or identity to each qubit in the circuit. We then define $\vec{P}(\alpha) := \exp(-i\frac\alpha2 \vec{P})$ to be its associated \emph{exponentiated Pauli} gate. For instance, $Z_i(\alpha)$ is the standard $Z$ phase gate (up to global phase) over an angle $\alpha$ applied to qubit $i$. For instance, the $S$ gate is given by $S_i := Z_i(\frac\pi2)$ and the $T$ gate is $T_i := Z_i(\frac\pi4)$.

Given any unitary quantum circuit consisting of Clifford gates and phase gates over $Z$ we can efficiently transform it into a circuit consisting of just exponentiated Pauli gates, see for instance \cite{litinski2019game} or \cite{zhang2019optimizing}). Let us give a quick overview of this procedure.

Starting from the beginning of the circuit, pick the first non-Clifford phase gate you encounter. We can view this phase gate as an exponentiated Pauli $Z_i(\alpha)$. If it is the first gate of the circuit we are done and move to the next non-Clifford gate. If it is not the first gate then there must a Clifford gate $C$ before it. Due to the definitional property of the Clifford gates $C^\dagger Z_i C = \vec{P}$ for some Pauli string $\vec{P}$ and hence we calculate 
$$Z_i(\alpha)C = CC^\dagger\exp(-i\frac\alpha2 Z_i)C = C\exp(-i\frac\alpha2 C^\dagger Z_iC) = C\exp(-i\frac\alpha2 \vec{P}).$$
There is now one less Clifford gate before the exponentiated Pauli. Repeating this procedure we bring the exponentiated Pauli all the way to the front. We then continue the procedure with the next non-Clifford phase.
The circuit resulting from this procedure consists of a series of exponentiated Pauli gates, followed by some Clifford circuit. We could then synthesise the Clifford circuit and write each CNOT,CZ, S and H gate as an exponentiated Pauli, but it will often be useful to do something else with these remaining Clifford gates.

\subsection{Global interactions}

The type of global interaction we assume is known as an `Ising-type' interaction, and can be implemented on ion trap quantum computers using the Mølmer-Sørensen (MS) interaction, which is why~\cite{maslov2018use} referred to these interactions as `global Mølmer-Sørensen' (GMS) gates. The `local' MS gate $XX$ acting on qubits $i$ and $j$ is given by the exponentiated Pauli $XX_{ij}(\alpha) = \exp(-i \frac{\alpha}{2} X_iX_j)$. Here $\alpha$ is a parameter that can be modified in the experimental apparatus. The global MS gate acting on a set of qubits $S$ is given by 
\begin{equation}\label{eq:GMS}
\GMS_S(\alpha) = \exp(-i \frac{\alpha}{4}\sum_{i\neq j\in S} X_iX_j) = \prod_{i\neq j\in S} \exp(-i \frac{\alpha}{4} X_iX_j).
\end{equation}

This last equality follows because all the gates $X_iX_j$ commute with one another. Note that we use $\frac{\alpha}{4}$ instead of $\frac{\alpha}{2}$ to compensate for the double counting of indices, as the MS gate is symmetric in the two qubits it acts on. 
Eq.~\eqref{eq:GMS} shows that the global MS gate corresponds to applying a local MS gate to every pair of qubits it acts on.
Note that in \cite{maslov2018use} they allow each pair of qubits to have a different interaction strength $\alpha$. We don't consider that possibility in this paper.

For us it will be convenient to consider a related gate. When we apply a Hadamard gate before and after every qubit that the GMS gate acts on, the Pauli $X$ is changed to a Pauli $Z$: $\GZZ_S(\alpha) := H^{\otimes n}~\GMS_S(\alpha)~H^{\otimes n} = \prod_{i< j\in S} ZZ_{i,j}(\alpha)$.
As is shown in for instance \cite{kissinger2017MBQC}, we can easily transform $ZZ_{i,j}(\frac\pi2)$ into a CZ gate (up to global phase): $\CZ_{i,j} = S^\dagger_i S^\dagger_j ZZ_{i,j}(\frac\pi2)$, where the S gate is the $Z(\frac\pi2)$ rotation.

Combining these observations, we see that the $\GMS_S(\frac\pi2)$ gate can be transformed by local Cliffords into an application of a CZ gate to every pair of qubits in the set $S$. Let us call this gate $\GCZ_S$, for `global CZ gate'. Hence, for the purposes of synthesising circuits where the metric is the number of \GMS gates, we can equivalently work with \GZZ or \GCZ gates, at the cost of potentially introducing additional local Clifford gates (that in most cost models for circuits are considered cheap).

We will make a distinction between two different types of global interactions. The ones described above we will call \emph{targeted}, because we can apply them to any subset of the qubits of the device. In Section~\ref{sec:untargeted} we will also consider \emph{untargeted} global interactions. These are gates that always interact on all the qubits at once. Hence, more local 2-qubit interactions will have to be implemented using some more clever methods. If we don't specify which type of interaction is meant, a global interaction is understood to be targeted.

\section{Synthesising Clifford circuits using targeted global gates}\label{sec:clifford}

A naive way to synthesise a Clifford circuit using global gates is to synthesise it using regular CNOT gates and single-qubit unitaries and constructing each CNOT gate using a global gate. As normal forms of $n$-qubit Clifford unitaries contains $O(n^2/\log(n))$ CNOT gates~\cite{markov2008optimal}, this approach would also require $O(n^2/\log(n))$ GMS gates.

As was shown in \cite{maslov2018use}, a better strategy is to synthesise an entire layer of CNOT gates at once. This allows one to take advantage of the efficient implementation of what is called a \emph{fan-out gate} in \cite{maslov2018use}, a series of CNOT gates with a common control qubit:
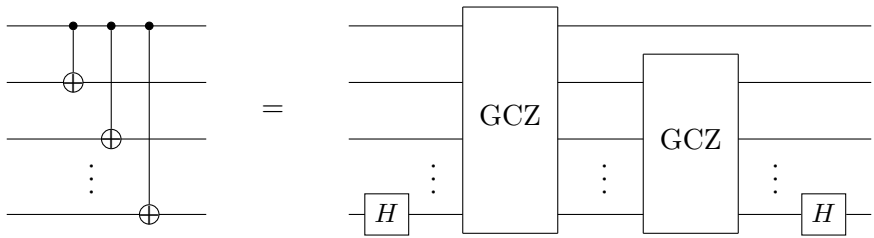
\new{
\ctikzfig{fan-out}
\emph{Fan-in gates}, where the CNOT gates instead have a common target, can be constructed in a similar way.
}

Using fan-out gates, it was shown in \cite{maslov2018use} that an arbitrary $n$-qubit linear Boolean circuit, i.e.~a CNOT circuit, can be synthesised using just $4n-6$ GMS gates. 
\new{
    Note that this is an almost quadratic improvement over the $O(n^2/\log(n))$ CNOT gates required to synthesise an arbitrary CNOT circuit.
}
In combination with the CPCPHPCPC normal form for Clifford circuits of~\cite{aaronsongottesman2004} that consists of single qubit gates and four layers of CNOT gates, this results in a synthesis strategy for arbitrary Clifford circuits requiring $12n-18$ GMS gates.%
\footnote{
\new{Two of the CNOT layers are arbitrary and hence require $4n-6$ GMS gates each. However, the other two CNOT layers correspond to upper triangular Boolean matrices and require only $2n-3$ each.}
}

There are however also other Clifford normal forms, and some of those result in a more efficient implementation. 
\new{
This is especially the case when the normal form requires CZ layers instead of CNOT layers, such as the Clifford normal form presented in~\cite{cliffsimp} which consists of layers H-S-CZ-CNOT-H-CZ-S-H.
}
These CZ layers can be implemented using GCZ gates (and thus GMS gates) even more efficiently than a CNOT layer:
\begin{proposition}\label{prop:CZ-construction}
	An $n$-qubit circuit consisting of just CZ gates can be implemented using local Clifford gates and at most $n-1$ targeted GMS gates.
\end{proposition}
\begin{proof}
We can construct a GCZ gate out a GMS gate using local Cliffords, so it suffices to show the claim for GCZ gates.
Suppose given an $n$-qubit CZ circuit. Pick an ordering on those qubits involved in at least one CZ gate in this circuit. We may assume there is at most one CZ gate between any pair of qubits, since otherwise pairs of gates would cancel.

Let $S_1$ be the set of qubits that have a CZ gate to qubit $1$. 
\new{
Apply a $\GCZ_{S_1\cup \{1\}}$ gate to this circuit. This might introduce additional CZ gates between the qubits of $S_1$, but crucially, it cancels all the CZ gates between qubit $1$ and the qubits of $S_1$. 
}
Hence, the remaining CZ gates do not involve qubit $1$. 

Now, in the resulting circuit, let $S_2$ be the set of qubits that have a CZ gate to qubit $2$. 
Apply a $\GCZ_{S_2\cup \{2\}}$ gate to this circuit. Again, this cancels all CZ gates between qubit $2$ and the other qubits. Note that this does not introduce new CZ gates to qubit $1$. Hence, by repeating this procedure until we get to qubit $n-1$, in which case the only possible remaining CZ gate is $\CZ_{n-1,n}$, we will have succeeded in canceling all remaining CZ gates, and we are left with the identity circuit.
Hence, our sequence of GCZ gates is the inverse of the CZ circuit we started with.
As CZ circuits are self-inverse, the series of GCZ gates we have just constructed is an implementation of the CZ circuit.
\end{proof}

\new{
For an example of this procedure see Figure~\ref{fig:CZ-circuit}. Note that an arbitrary CZ circuit can consist of up to $n^2/2$ CZ gates, so this proposition again shows a quadratic improvement in using GMS gates.

\begin{figure}
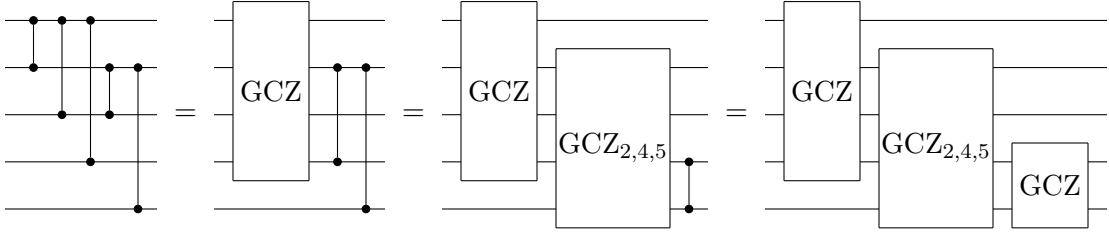

\ctikzfig{CZ-circuit}
\caption{A CZ circuit and its synthesis using GCZ gates.}
\label{fig:CZ-circuit}
\end{figure}
}

Combining the synthesis of CZ circuits and CNOT circuits using GMS gates we get the following improved translation of arbitrary Clifford circuits into GMS circuits:

\begin{theorem}\label{thm:synth-clifford}
	An $n$-qubit Clifford unitary can be implemented using local Clifford gates and at most $6n-8$ GMS gates.
\end{theorem}
\begin{proof}
	In~\cite{cliffsimp}, a normal form for Clifford circuits is found that contains two layers of CZ gates and a single layer of CNOTs (and all other layers requiring only single-qubit Clifford unitaries). The CNOT layer requires at most $4n-6$ GMS gates~\cite{maslov2018use} and each CZ layer requires at most $n-1$ GMS gates by the previous proposition. 
\end{proof}

To the authors knowledge, every other normal form for Clifford circuits contains at least two CNOT layers, and hence, this is the current best bound on the number of GMS gates needed to implement a Clifford unitary.

\section{Constructing arbitrary circuits using targeted global gates}\label{sec:constructingCircuit}

In this section we will show how an arbitrary unitary circuit can be synthesised using targeted \GMS gates and single-qubit phase gates. We assume the unitary we wish to synthesise is given in the universal gate set consisting of Clifford gates and arbitrary single qubit Z-rotations. We will refer to this gate set as the Clifford+Phases gate set. Concretely, the given circuit can consist of the Clifford gates S, Hadamard, CNOT and CZ, supplemented with $Z(\alpha)$ gates for arbitrary $\alpha\in [0,2\pi]$. Our synthesis will contain a number of \GMS gates that is directly proportional to the number of non-Clifford phase gates, i.e.~the number of $Z(\alpha)$ gates for which $\alpha\neq k\frac\pi2$ for all $k\in \mathbb{Z}$.
As \GCZ gates are equivalent up to local Clifford gates to $\GMS(\frac\pi2)$ gates we are warranted in using \GCZ gates in our synthesis instead of \GMS gates, as in the previous section.

As shown in Section~\ref{sec:expPauli}, a Clifford+Phases circuits can be easily transformed into a series of exponentiated Pauli gates. Let us therefore first see how to construct an arbitrary exponentiated Pauli using GCZ gates and single qubit gates. 

Let $\vec{P}$ be some Pauli string. We can always find a circuit of local Cliffords $C_1\otimes \cdots \otimes C_n$ such that $(C_1\otimes \cdots \otimes C_n)\vec{P}(C_1\otimes \cdots \otimes C_n)^\dagger = \vec{P}'$ where $\vec{P}'$ consists of just identities and $Z$ gates: $\vec{P}' = Z_{i_1}Z_{i_2}\cdots Z_{i_k}$.
Hence:
$$\vec{P}(\alpha) := \exp(-i\frac\alpha2 \vec{P}) = (C_1\otimes \cdots \otimes C_n)^\dagger\exp(-i\frac\alpha2 \vec{P}')(C_1\otimes \cdots \otimes C_n).$$
To synthesise an exponentiated Pauli it then suffices to show how to construct an exponentiated Pauli containing only Z's using GCZ gates and single-qubit gates.
Let $S$ be the set of qubits on which $\vec{P}'$ has a Z. Then, up to global phase, $\vec{P}'(\alpha)$ acts by applying a phase $e^{i\alpha}$ iff the parity of the qubits in the set $S$ is odd, and acts as the identity otherwise. 
This type of diagonal gate was dubbed a \emph{phase gadget} in \cite{kissinger2019tcount}, and they play an important role in the theory of \emph{phase polynomials}~\cite{amy2014polynomial}.

There are various ways to realise phase gadgets with common gate sets (see for instance \cite{phaseGadgetSynth}). These constructions work by building a `ladder of CNOTs' to achieve the appropriate parity on a qubit, applying a phase gate on this parity, and then constructing the opposite ladder to bring the qubits back to their original state. Our construction also consists of a ladder, but a different one than the one commonly used; see Figure~\ref{fig:phase-gadget-ladder}.

\begin{figure}[!htb]
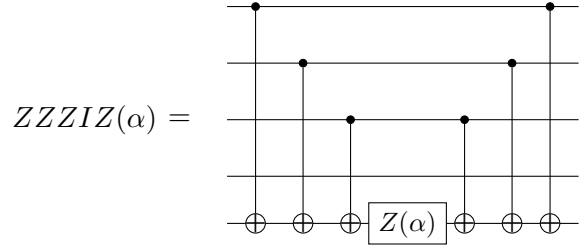

	\centering
	\ctikzfig{phase-gadget-ladder}
	\caption{Circuit realisation of the exponentiated Pauli $ZZZIZ(\alpha)$. Note that none of the gates interact with the 4th qubit as expected (as the Pauli acts trivially there).}
	\label{fig:phase-gadget-ladder}
\end{figure}

Note that while a phase gadget is symmetric in the qubits it acts on, the circuit realisation is not, and thus there is some freedom in choosing how to construct the CNOT ladders. In particular, we can choose which qubit is the target of the ladder. For the results of this paper this won't be important, and we pick the qubit arbitrarily, but it might be useful in NISQ architectures where certain qubits might have better interaction properties.

Such a ladder of CNOT gates is called a \emph{fan-in gate} in~\cite{maslov2018use}, and they showed that it can be constructed using 2 GMS gates. In~\cite{yung2014transistor} it was realised that in fact the complete phase gadget can be synthesised using just 2 GMS gates.
For our purposes it will however be sufficient to construct fan-in gates up to a collection of 'garbage' Clifford gates, which turns out to be possible using just a single GMS gate (or equivalently, a single \GCZ gate).

For instance, if we wish to construct the exponentiated Pauli gate of Figure~\ref{fig:phase-gadget-ladder} using a \GCZ gate, we can do the following:
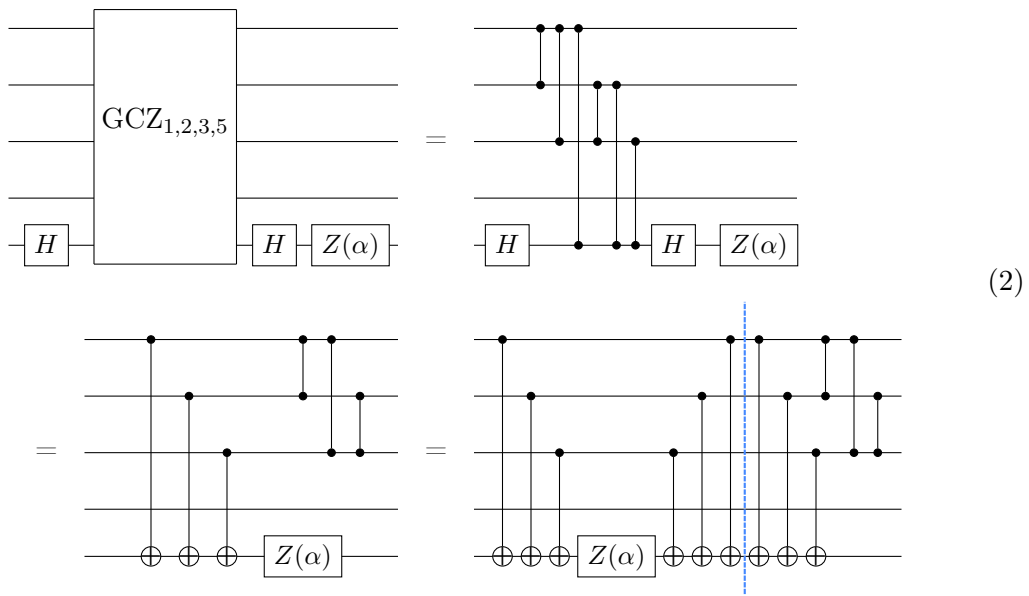
\begin{equation}\label{eq:phase-gadget-GCZ}
\begin{tikzpicture}
	\begin{pgfonlayer}{nodelayer}
		\node [style=none] (0) at (1, -1.5) {};
		\node [style=none] (1) at (12.25, -1.5) {};
		\node [style=none] (2) at (1, -3) {};
		\node [style=none] (3) at (12.25, -3) {};
		\node [style=none] (4) at (1, -4.5) {};
		\node [style=none] (5) at (12.25, -4.5) {};
		\node [style=vertex] (9) at (2.5, -3) {};
		\node [style=vertex set] (10) at (2.5, -7.25) {$\!\!+\!\!$};
		\node [style=vertex set] (11) at (3.25, -7.25) {$\!\!+\!\!$};
		\node [style=vertex] (12) at (3.25, -4.5) {};
		\node [style=box] (17) at (4.75, -7.25) {$Z(\alpha)$};
		\node [style=none] (18) at (1, -6) {};
		\node [style=none] (19) at (12.25, -6) {};
		\node [style=vertex] (20) at (7, -3) {};
		\node [style=vertex set] (21) at (7, -7.25) {$\!\!+\!\!$};
		\node [style=vertex set] (22) at (6.25, -7.25) {$\!\!+\!\!$};
		\node [style=vertex] (23) at (6.25, -4.5) {};
		\node [style=vertex] (24) at (1.75, -1.5) {};
		\node [style=vertex set] (25) at (1.75, -7.25) {$\!\!+\!\!$};
		\node [style=vertex] (26) at (7.75, -1.5) {};
		\node [style=vertex set] (27) at (7.75, -7.25) {$\!\!+\!\!$};
		\node [style=none] (28) at (1, -7.25) {};
		\node [style=none] (29) at (12.25, -7.25) {};
		\node [style=none] (31) at (0, 3.75) {=};
		\node [style=none] (32) at (-11.25, 6.75) {};
		\node [style=none] (33) at (-1, 6.75) {};
		\node [style=none] (34) at (-11.25, 5.25) {};
		\node [style=none] (35) at (-1, 5.25) {};
		\node [style=none] (36) at (-11.25, 3.75) {};
		\node [style=none] (37) at (-1, 3.75) {};
		\node [style=none] (43) at (-11.25, 2.25) {};
		\node [style=none] (44) at (-1, 2.25) {};
		\node [style=none] (53) at (-11.25, 1) {};
		\node [style=none] (54) at (-1, 1) {};
		\node [style=none] (55) at (-9, 7.25) {};
		\node [style=none] (56) at (-9, 0.5) {};
		\node [style=none] (57) at (-5.25, 0.5) {};
		\node [style=none] (58) at (-5.25, 7.25) {};
		\node [style=none] (59) at (-9, 6.75) {};
		\node [style=none] (60) at (-9, 5.25) {};
		\node [style=none] (61) at (-9, 3.75) {};
		\node [style=none] (62) at (-9, 2.25) {};
		\node [style=none] (63) at (-9, 1) {};
		\node [style=none] (64) at (-5.25, 1) {};
		\node [style=none] (65) at (-5.25, 2.25) {};
		\node [style=none] (66) at (-5.25, 3.75) {};
		\node [style=none] (67) at (-5.25, 5.25) {};
		\node [style=none] (68) at (-5.25, 6.75) {};
		\node [style=none] (69) at (-7.125, 4.25) {$\GCZ_{1,2,3,5}$};
		\node [style=box] (70) at (-10.25, 1) {$H$};
		\node [style=box] (71) at (-4.25, 1) {$H$};
		\node [style=box] (72) at (-2.25, 1) {$Z(\alpha)$};
		\node [style=none] (74) at (1, 6.75) {};
		\node [style=none] (75) at (9.5, 6.75) {};
		\node [style=none] (76) at (1, 5.25) {};
		\node [style=none] (77) at (9.5, 5.25) {};
		\node [style=none] (78) at (1, 3.75) {};
		\node [style=none] (79) at (9.5, 3.75) {};
		\node [style=vertex] (80) at (4.75, 5.25) {};
		\node [style=vertex] (83) at (5.25, 3.75) {};
		\node [style=box] (84) at (8.5, 1) {$Z(\alpha)$};
		\node [style=none] (85) at (1, 2.25) {};
		\node [style=none] (86) at (9.5, 2.25) {};
		\node [style=vertex] (91) at (3.75, 6.75) {};
		\node [style=none] (95) at (1, 1) {};
		\node [style=none] (96) at (9.5, 1) {};
		\node [style=vertex] (97) at (2.75, 6.75) {};
		\node [style=vertex] (98) at (3.75, 1) {};
		\node [style=vertex] (99) at (4.75, 1) {};
		\node [style=vertex] (100) at (5.25, 1) {};
		\node [style=box] (101) at (1.875, 1) {$H$};
		\node [style=box] (102) at (6.25, 1) {$H$};
		\node [style=vertex] (103) at (2.75, 5.25) {};
		\node [style=vertex] (104) at (4.25, 5.25) {};
		\node [style=vertex] (105) at (4.25, 3.75) {};
		\node [style=vertex] (106) at (3.25, 3.75) {};
		\node [style=vertex] (107) at (3.25, 6.75) {};
		\node [style=none] (108) at (0, -4.5) {=};
		\node [style=none] (109) at (8.125, -0.5) {};
		\node [style=none] (110) at (8.125, -8.25) {};
		\node [style=vertex] (111) at (9.25, -3) {};
		\node [style=vertex set] (112) at (9.25, -7.25) {$\!\!+\!\!$};
		\node [style=vertex set] (113) at (10, -7.25) {$\!\!+\!\!$};
		\node [style=vertex] (114) at (10, -4.5) {};
		\node [style=vertex] (115) at (8.5, -1.5) {};
		\node [style=vertex set] (116) at (8.5, -7.25) {$\!\!+\!\!$};
		\node [style=vertex] (117) at (11.625, -4.5) {};
		\node [style=vertex] (118) at (11.625, -3) {};
		\node [style=vertex] (119) at (10.25, -3) {};
		\node [style=vertex] (120) at (10.25, -1.5) {};
		\node [style=vertex] (121) at (11, -4.5) {};
		\node [style=vertex] (122) at (11, -1.5) {};
		\node [style=none] (123) at (-10.25, -4.5) {=};
		\node [style=none] (124) at (-9.25, -1.5) {};
		\node [style=none] (125) at (-1, -1.5) {};
		\node [style=none] (126) at (-9.25, -3) {};
		\node [style=none] (127) at (-1, -3) {};
		\node [style=none] (128) at (-9.25, -4.5) {};
		\node [style=none] (129) at (-1, -4.5) {};
		\node [style=vertex] (130) at (-6.5, -3) {};
		\node [style=vertex] (131) at (-5.5, -4.5) {};
		\node [style=box] (132) at (-3.5, -7.25) {$Z(\alpha)$};
		\node [style=none] (133) at (-9.25, -6) {};
		\node [style=none] (134) at (-1, -6) {};
		\node [style=vertex] (135) at (-7.5, -1.5) {};
		\node [style=none] (136) at (-9.25, -7.25) {};
		\node [style=none] (137) at (-1, -7.25) {};
		\node [style=vertex] (138) at (-3.5, -1.5) {};
		\node [style=vertex] (144) at (-3.5, -3) {};
		\node [style=vertex] (145) at (-2, -3) {};
		\node [style=vertex] (146) at (-2, -4.5) {};
		\node [style=vertex] (147) at (-2.75, -4.5) {};
		\node [style=vertex] (148) at (-2.75, -1.5) {};
		\node [style=vertex set] (149) at (-7.5, -7.25) {$\!\!+\!\!$};
		\node [style=vertex set] (150) at (-6.5, -7.25) {$\!\!+\!\!$};
		\node [style=vertex set] (151) at (-5.5, -7.25) {$\!\!+\!\!$};
	\end{pgfonlayer}
	\begin{pgfonlayer}{edgelayer}
		\draw [in=180, out=0] (0.center) to (1.center);
		\draw (3.center) to (2.center);
		\draw [in=180, out=0] (4.center) to (5.center);
		\draw (10) to (9);
		\draw (11) to (12);
		\draw [in=180, out=0] (18.center) to (19.center);
		\draw (21) to (20);
		\draw (22) to (23);
		\draw (25) to (24);
		\draw (27) to (26);
		\draw [in=180, out=0] (28.center) to (29.center);
		\draw (55.center) to (58.center);
		\draw (58.center) to (57.center);
		\draw (57.center) to (56.center);
		\draw (56.center) to (55.center);
		\draw (68.center) to (33.center);
		\draw (35.center) to (67.center);
		\draw (66.center) to (37.center);
		\draw (44.center) to (65.center);
		\draw (64.center) to (54.center);
		\draw (53.center) to (63.center);
		\draw (62.center) to (43.center);
		\draw (36.center) to (61.center);
		\draw (60.center) to (34.center);
		\draw (32.center) to (59.center);
		\draw [in=180, out=0] (74.center) to (75.center);
		\draw (77.center) to (76.center);
		\draw [in=180, out=0] (78.center) to (79.center);
		\draw [in=180, out=0] (85.center) to (86.center);
		\draw [in=180, out=0] (95.center) to (96.center);
		\draw (98) to (91);
		\draw (99) to (80);
		\draw (100) to (83);
		\draw (103) to (97);
		\draw (105) to (104);
		\draw (106) to (107);
		\draw [style=hadamard edge] (110.center) to (109.center);
		\draw (112) to (111);
		\draw (113) to (114);
		\draw (116) to (115);
		\draw (117) to (118);
		\draw (120) to (119);
		\draw (122) to (121);
		\draw [in=180, out=0] (124.center) to (125.center);
		\draw (127.center) to (126.center);
		\draw [in=180, out=0] (128.center) to (129.center);
		\draw [in=180, out=0] (133.center) to (134.center);
		\draw [in=180, out=0] (136.center) to (137.center);
		\draw (144) to (138);
		\draw (146) to (145);
		\draw (147) to (148);
		\draw (149) to (135);
		\draw (150) to (130);
		\draw (151) to (131);
	\end{pgfonlayer}
\end{tikzpicture}
\end{equation}

Here the connected black dots represent CZ gates. The second equality uses the fact that CZ gates commute past one another and that conjugating one of the qubits of a CZ gate with Hadamards transforms the CZ gate into a CNOT. The last equality follows by introducing pairs of canceling CNOT gates. The blue dotted wire represents the dividing line between the implementation of the exponentiated Pauli and the `garbage' Clifford gates.

In general, to construct a phase gadget which acts non-trivially on the qubit set $S$, we simply apply a $\GCZ_S$ gate, conjugated on one (arbitrarily chosen) qubit by Hadamard gates, and then followed on the same qubit by a $Z(\alpha)$ gate, where $\alpha$ matches the phase of the phase gadget.

We now have all the necessary ingredients for our construction; see Algorithm~\ref{alg:general-circuit}.

\begin{algorithm}
\caption{Transforming a general circuit into one using targeted global gates.}
\label{alg:general-circuit}
\begin{enumerate}
\item Given a Clifford+Phases circuit, take the earliest non-Clifford phase and push it the beginning of the circuit, transforming it into an exponentiated Pauli gate, as described in Section~\ref{sec:expPauli}.

\item Push out local Clifford gates so that the Pauli string only contains $Z$'s, making it a phase gadget.

\item Synthesise the phase gadget up to some `garbage' Clifford gates using a single GCZ gate as described above.

\item If the remaining circuit contains no non-Clifford phase gates, go to the next step. Otherwise, take the earliest non-Clifford phase and push it to the left past the earlier (necessarily Clifford) gates in the circuit and the `garbage' gates arising from the synthesis of the previous phase gadgets. Repeat steps 2 and 3 for the resulting exponentiated Pauli.
\item The remainder of the circuit does not contain any non-Clifford phase gates, and hence is a Clifford circuit. We can synthesise this circuit using the methods described in Section~\ref{sec:clifford}.
\end{enumerate}
\end{algorithm}
Note that each iteration of steps 2-4 consumes a non-Clifford phase gate, and hence this algorithm indeed terminates.

Combining this algorithm with the synthesis of Clifford circuits described in the previous section we get the following result:
\begin{theorem}\label{thm:circuit-with-GMS}
	A $n$-qubit circuit consisting of Clifford gates and $N$ non-Clifford Z-phase gates can be implemented using single qubit unitaries and at most $N+6n-8$ GMS gates.
\end{theorem}
\begin{proof}
	Algorithm~\ref{alg:general-circuit} shows how the non-Clifford portion of the circuit can be implemented using $N$ GMS gates. The remaining Clifford circuit then requires an additional $6n-8$ GMS gates by Theorem~\ref{thm:synth-clifford}.
\end{proof}

\new{
For many circuits of interest, the number of non-Clifford phases will be significantly lower than the number of CNOT gates required, see for instance the optimised benchmark circuits of~\cite{nam2018automated} or the quantum chemistry circuits of~\cite{cowtan_generic_2020}, and hence this result shows there is a benefit to using targeted GMS gates over CNOT gates.
For certain specific circuit constructions the advantage can however be less clear cut.
For instance, in~\cite{heDecompositionsNqubitToffoli2017} a construction for the $n$-controlled Toffoli gate using just one dirty ancilla requires $24n-72$ CNOT gates, but $32n-96$ T gates, which would mean that our algorithm would give no benefit for this construction (although, this might change if T-count optimisation algorithms are applied prior to synthesis). 
Using the Toffoli construction of~\cite{maslov2016advantages} which requires more ancillae uses just $6n-12$ CNOT gates and $8n-17$ T gates, which brings the counts closer. 
Using handcrafted circuits, in~\cite{maslov2018use} they find a construction of an $n$-controlled Toffoli using GMS gates which requires just $3n-6$ GMS gates (although requiring the use of varying interaction strengths in the GMS gate), so it is still possible to get a benefit using GMS gates even for constructing Toffoli gates.
}

\section{Constructing circuits from untargeted global gates}\label{sec:untargeted}

In some devices it might not be possible to target `global' gates to just a subset of the total number of qubits, and instead the global gate must always be applied to \emph{all} qubits at once.

In \cite{maslov2018use} it was shown how to construct an $(n-1)$-qubit $\GMS(\alpha)$ gate using two copies of $n$-qubit $\GMS(\frac\alpha2)$ gates. By iterating the procedure, a $(n-k)$-qubit $\GMS(\alpha)$ gate can be constructed by using $2^k$ $\GMS(\frac{\alpha}{2^k})$ gates. This exponential increase in the cost when making the gate more local is obviously undesirable. In this section we describe a modification to their construction that scales linearly as the gate becomes more local.

Instead of working with GMS gates, we will work with the (local-Clifford equivalent) GZZ gates, as these are diagonal and correspond to phase polynomials that are more convenient to work with. In order to make the calculations legible we will use ZX-diagrams~\cite{CD1,CD2}. No previous understanding on the ZX-calculus will be assumed, but we direct the interested reader to the book~\cite{CKbook} or the review~\cite{van_de_wetering_zx-calculus_2020}.

First let us recall and prove the construction used in \cite{maslov2018use} to make a targeted global gate out of a pair of untargeted ones, but now presented using ZX-diagrams. Recall that a $\GZZ_S(\alpha)$ gate consists of $ZZ(\alpha)$ gates applied to every pair of qubits in $S$.
A $ZZ(\alpha)$ gate, i.e.~a phase gadget, can be represented in the ZX-calculus in the following way~\cite{kissinger2017MBQC}:
\[ZZ(\alpha) \ \ = \ \ \begin{tikzpicture}
	\begin{pgfonlayer}{nodelayer}
		\node [style=none] (0) at (-8.25, -1) {};
		\node [style=none] (1) at (-5.5, -1) {};
		\node [style=none] (2) at (-8.25, 1) {};
		\node [style=none] (3) at (-5.5, 1) {};
		\node [style=Z dot] (4) at (-7, 1) {};
		\node [style=Z dot] (5) at (-7, -1) {};
		\node [style=X dot] (6) at (-7, 0) {};
		\node [style=Z phase dot] (7) at (-8, 0) {$\alpha$};
	\end{pgfonlayer}
	\begin{pgfonlayer}{edgelayer}
		\draw (0.center) to (5);
		\draw (5) to (1.center);
		\draw (5) to (6);
		\draw (6) to (4);
		\draw (4) to (2.center);
		\draw (4) to (3.center);
		\draw (7) to (6);
	\end{pgfonlayer}
\end{tikzpicture}\]

Note that here the orientation of the middle nodes is irrelevant:
\ctikzfig{phase-gadget-orientation}
The commutativity of GZZ gates can be expressed as follows:
\ctikzfig{phase-gadget-commute}
Those not familiar with ZX-diagrams can simply view the middle diagram where the dots are `fused' as a shorthand for the commutativity of these gates.

In this section we will only need the following identities that hold for the $ZZ(\alpha)$ gate:
\begin{equation}\label{eq:phase-gadget-identities}
\begin{tikzpicture}
	\begin{pgfonlayer}{nodelayer}
		\node [style=none] (149) at (-13.5, 2.25) {};
		\node [style=none] (150) at (-13.5, 0.75) {};
		\node [style=none] (152) at (-8.75, 2.25) {};
		\node [style=none] (153) at (-8.75, 0.75) {};
		\node [style=Z dot] (155) at (-11.75, 2.25) {};
		\node [style=Z dot] (156) at (-11.75, 0.75) {};
		\node [style=X dot] (160) at (-11.25, 1.5) {};
		\node [style=Z phase dot] (162) at (-10.5, 1.5) {$\alpha$};
		\node [style=X phase dot] (163) at (-9.5, 2.25) {$\pi$};
		\node [style=X phase dot] (164) at (-12.75, 2.25) {$\pi$};
		\node [style=none] (165) at (-6, 2.25) {};
		\node [style=none] (166) at (-6, 0.75) {};
		\node [style=none] (167) at (-2.75, 2.25) {};
		\node [style=none] (168) at (-2.75, 0.75) {};
		\node [style=Z dot] (169) at (-5, 2.25) {};
		\node [style=Z dot] (170) at (-5, 0.75) {};
		\node [style=X dot] (171) at (-4.5, 1.5) {};
		\node [style=Z phase dot] (172) at (-3.5, 1.5) {\ $-\alpha~$};
		\node [style=none] (175) at (-7.25, 1.5) {=};
		\node [style=none] (176) at (-13.5, -0.5) {};
		\node [style=none] (177) at (-13.5, -2) {};
		\node [style=none] (178) at (-8.75, -0.5) {};
		\node [style=none] (179) at (-8.75, -2) {};
		\node [style=Z dot] (180) at (-11.75, -0.5) {};
		\node [style=Z dot] (181) at (-11.75, -2) {};
		\node [style=X dot] (182) at (-11.25, -1.25) {};
		\node [style=Z phase dot] (183) at (-10.5, -1.25) {$\alpha$};
		\node [style=X phase dot] (184) at (-9.5, -0.5) {$\pi$};
		\node [style=X phase dot] (185) at (-12.75, -0.5) {$\pi$};
		\node [style=none] (186) at (-6, -0.5) {};
		\node [style=none] (187) at (-6, -2) {};
		\node [style=none] (188) at (-2.75, -0.5) {};
		\node [style=none] (189) at (-2.75, -2) {};
		\node [style=Z dot] (190) at (-5, -0.5) {};
		\node [style=Z dot] (191) at (-5, -2) {};
		\node [style=X dot] (192) at (-4.5, -1.25) {};
		\node [style=Z phase dot] (193) at (-3.75, -1.25) {$\alpha$};
		\node [style=none] (194) at (-7.25, -1.25) {=};
		\node [style=X phase dot] (195) at (-9.5, -2) {$\pi$};
		\node [style=X phase dot] (196) at (-12.75, -2) {$\pi$};
		\node [style=none] (197) at (1.25, 2.25) {};
		\node [style=none] (198) at (1.25, 0.75) {};
		\node [style=none] (199) at (6, 2.25) {};
		\node [style=none] (200) at (6, 0.75) {};
		\node [style=Z dot] (201) at (2.25, 2.25) {};
		\node [style=Z dot] (202) at (2.25, 0.75) {};
		\node [style=X dot] (203) at (2.75, 1.5) {};
		\node [style=Z phase dot] (204) at (3.5, 1.5) {$\alpha$};
		\node [style=Z dot] (205) at (4.25, 2.25) {};
		\node [style=Z dot] (206) at (4.25, 0.75) {};
		\node [style=X dot] (207) at (4.75, 1.5) {};
		\node [style=Z phase dot] (208) at (5.5, 1.5) {$\beta$};
		\node [style=none] (209) at (7.25, 1.5) {=};
		\node [style=none] (210) at (8.25, 2.25) {};
		\node [style=none] (211) at (8.25, 0.75) {};
		\node [style=none] (212) at (12.25, 2.25) {};
		\node [style=none] (213) at (12.25, 0.75) {};
		\node [style=Z dot] (214) at (9.25, 2.25) {};
		\node [style=Z dot] (215) at (9.25, 0.75) {};
		\node [style=X dot] (216) at (9.75, 1.5) {};
		\node [style=Z phase dot] (217) at (11.25, 1.5) {\ $\alpha+\beta~$};
		\node [style=none] (218) at (1.25, -0.5) {};
		\node [style=none] (219) at (1.25, -2) {};
		\node [style=none] (220) at (4.5, -0.5) {};
		\node [style=none] (221) at (4.5, -2) {};
		\node [style=Z dot] (222) at (2.25, -0.5) {};
		\node [style=Z dot] (223) at (2.25, -2) {};
		\node [style=X dot] (224) at (2.75, -1.25) {};
		\node [style=Z phase dot] (225) at (3.5, -1.25) {$0$};
		\node [style=none] (226) at (7.25, -1.25) {=};
		\node [style=none] (227) at (8.5, -0.5) {};
		\node [style=none] (228) at (8.5, -2) {};
		\node [style=none] (229) at (11.75, -0.5) {};
		\node [style=none] (230) at (11.75, -2) {};
	\end{pgfonlayer}
	\begin{pgfonlayer}{edgelayer}
		\draw (153.center) to (150.center);
		\draw (149.center) to (152.center);
		\draw (155) to (160);
		\draw (160) to (156);
		\draw (160) to (162);
		\draw (168.center) to (166.center);
		\draw (165.center) to (167.center);
		\draw (169) to (171);
		\draw (171) to (170);
		\draw (171) to (172);
		\draw (179.center) to (177.center);
		\draw (176.center) to (178.center);
		\draw (180) to (182);
		\draw (182) to (181);
		\draw (182) to (183);
		\draw (189.center) to (187.center);
		\draw (186.center) to (188.center);
		\draw (190) to (192);
		\draw (192) to (191);
		\draw (192) to (193);
		\draw (200.center) to (198.center);
		\draw (197.center) to (199.center);
		\draw (201) to (203);
		\draw (203) to (202);
		\draw (203) to (204);
		\draw (205) to (207);
		\draw (207) to (206);
		\draw (207) to (208);
		\draw (213.center) to (211.center);
		\draw (210.center) to (212.center);
		\draw (214) to (216);
		\draw (216) to (215);
		\draw (216) to (217);
		\draw (221.center) to (219.center);
		\draw (218.center) to (220.center);
		\draw (222) to (224);
		\draw (224) to (223);
		\draw (224) to (225);
		\draw (230.center) to (228.center);
		\draw (227.center) to (229.center);
	\end{pgfonlayer}
\end{tikzpicture}
\end{equation}
Here the red $\pi$ node is the Pauli X (i.e.~a NOT gate).

Using these rewrite rules we can reproduce the calculation in~\cite{maslov2018use} to show how to construct a targeted GZZ gate from a pair of untargeted ones (which we here present using 3 qubits); see Figure~\ref{fig:global-to-local}.

\begin{figure}[!htb]
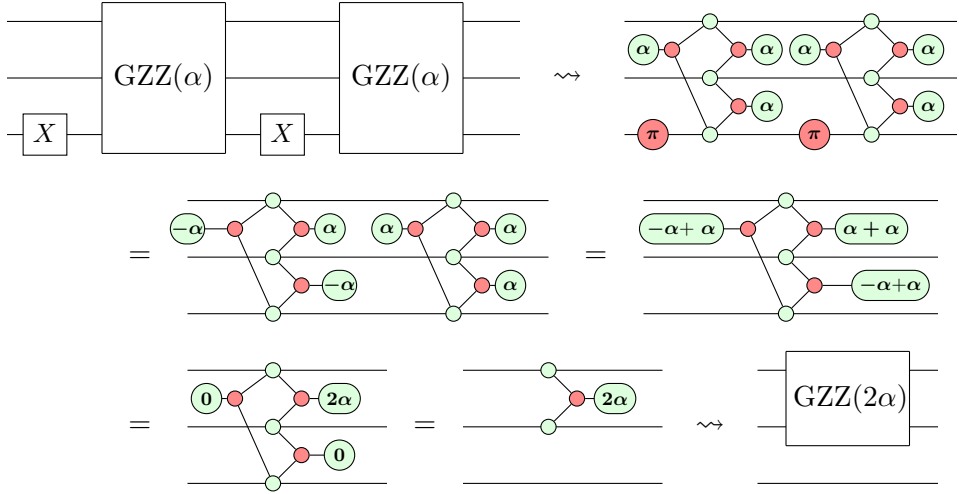

 \ctikzfig{global-to-local}
 \vspace{-0.2cm}
 \caption{Untargeted to targeted reduction from~\cite{maslov2018use} for 3 qubits.}
 \label{fig:global-to-local}
 \vspace{0.5cm}
\end{figure}


The first reduction step in this figure is just representing the circuit as a ZX-diagram, and all the other equalities use one of the identities in~\eqref{eq:phase-gadget-identities} together with the fact that $ZZ(\alpha)$ gates all commute with one another. The procedure for more qubits than 3 is entirely analogous.

To make GMS gates on lower number of qubits in~\cite{maslov2018use}, this procedure is simply repeated, removing one qubit line and doubling the number of GMS gates needed at every step. 
We will however adopt a modification of the procedure where we put $X$ gates on more than just one qubit at the same time. Let us demonstrate this with the simplest example on 4 qubits; see Figure~\ref{fig:4-to-2}.

\begin{figure}[!htb]
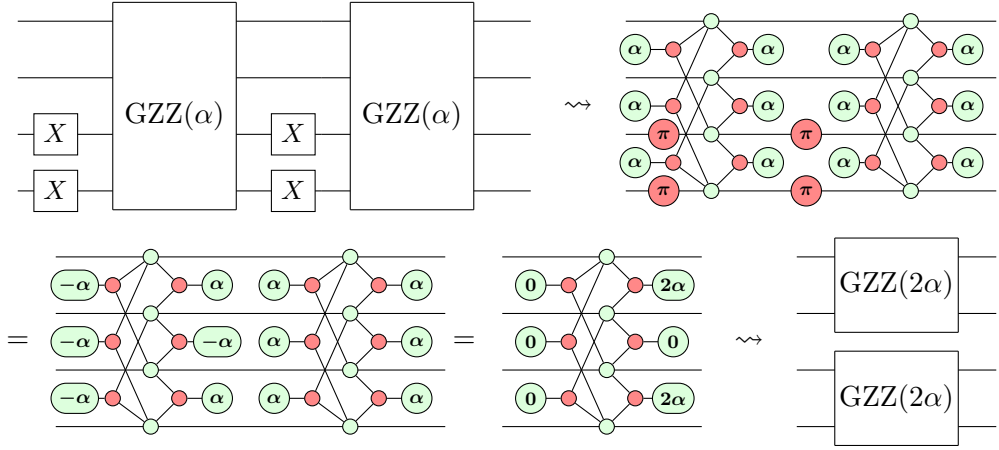

\ctikzfig{4-to-2}
\vspace{-0.2cm}
\caption{Untargeted to targeted reduction from 4 qubits to 2 qubits.}
\label{fig:4-to-2}
\end{figure}

The first step is again a translation from the circuit into a ZX-diagram. In the first equality note that for the $ZZ(\alpha)$ gate on the third and fourth qubit there are red $\pi$'s on both of its qubits and hence the phase is not flipped.
At the second equality we combine the gates, adding their phases together, resulting in the configuration of GZZ gates displayed.

This generalises in a straightforward way to a situation with more qubits and $X$ gates:
\begin{proposition}\label{prop:gzz-splits}
For any set of qubits $Q$ and subset of qubits $S\subseteq Q$ we have
$$\GZZ_Q(\alpha) \circ X_S \circ \GZZ_Q(\alpha)\circ X_S = \GZZ_{Q\backslash S}(2\alpha) \circ \GZZ_{S}(2\alpha),$$ 
where $X_S$ denotes the application of an $X$ gate on each of the qubits in the set $S$. In words: applying a pair of $\GZZ(\alpha)$ gates on a set $Q$ with NOT gates in between them on some of the qubits $S$ is equal to applying a $\GZZ(2\alpha)$ gate on each of the sets $Q\backslash S$ and $S$.
\end{proposition}
Note that for the purposes of this proposition, if $\lvert S\rvert = 1$ or $\lvert Q\backslash S\rvert = 1$, i.e.~if one of the sets of qubits contains just a single qubit, then the $\GZZ_S(2\alpha)$ gate, respectively $\GZZ_{Q\backslash S}(2\alpha)$ gate is just the identity (as is the case in the derivation of Figure~\ref{fig:global-to-local}).

Suppose the number of qubits is $2^n$. By taking $\lvert S \rvert = 2^{n-1}$ we get two parallel GZZ gates each acting on $2^{n-1}$ gates. If we start with four GZZ gates and apply this procedure to each pair we end with two parallel sets of GZZ gates, so that if we have the appropriate set of $X$'s in between them we can reduce these to four parallel GZZ gates each acting on $2^{n-2}$ gates. Applying this procedure $n-1$ times, which requires $2^{n-1}$ $\GZZ(\alpha)$ gates at the start, we end up with parallel $\GZZ(2^{n-1}\alpha)$ gates each acting on $2$ qubits each.

Now since conjugating a 2-qubit GZZ gate with a single pair of $X$ gates flips the phase of the GZZ gate, we can double the procedure one more time to cancel the 2-qubit GZZ gates we don't need by strategically inserting $X$ gates. Hence, by using $2^n$ $\GZZ(\alpha)$ gates acting on all $2^n$ qubits, we can implement any parallel set of $\GZZ(2^n\alpha)$ gates acting on 2 qubits. 

For example, in Figure~\ref{fig:8-to-2} we demonstrate the case where we have $n=3$, so that we have $8$ qubits, and we place the X gates so that we implement 2-qubit gates on the first and last pairs of qubits.
If we had used the procedure of~\cite{maslov2018use} to synthesise the same pair of 2-qubit GZZ gates, then each of the 2-qubit gates would have to be synthesised separately, and would require $2^{8-2} = 64$ \GZZ gates each. The benefits of this new procedure become even more pronounced as the number of qubits increases.
\begin{figure}
	\centering
	\scalebox{0.85}{\input{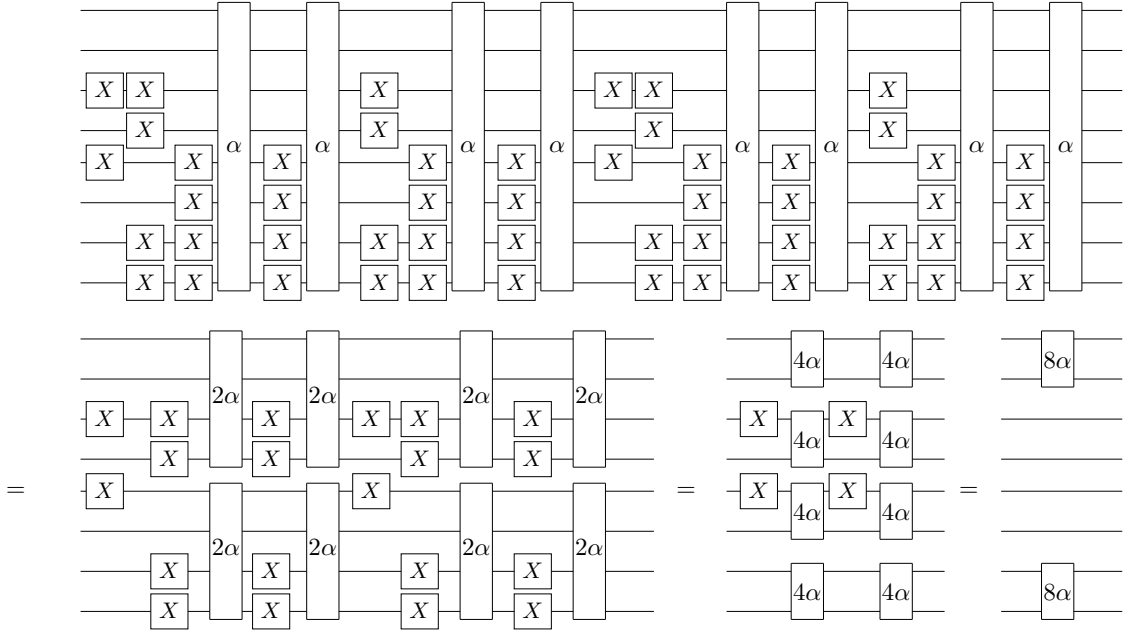}}
	\caption{Example of our procedure for constructing targeted GMS gates from untargeted ones. Here, the long boxes labelled by a phase represent $\GZZ$ gates. Each diagram in the equation is produced by applying Proposition~\ref{prop:gzz-splits} to pairs of GZZ gates.}
	\label{fig:8-to-2}
\end{figure}

If the number of qubits $n$ is not an exact power of $2$, the procedure still works, but requires a slightly different number of $n$-qubit GZZ gates. If $n=2^k + 1$, we require $2^k$ GZZ gates: the first pair splits into a $2^{k-1}$ qubit GZZ gate and a $2^{k-1}+1$ GZZ gate. We continue this splitting until we end up with a series of $2$-qubit GZZ gates and a 3-qubit GZZ gates. With one more doubling of the gate numbers we can cancel any of the $2$-qubit GZZ gates we want, and reduce the $3$-qubit GZZ gate to a 2-qubit one of our choice.

If instead $n=2^k + r$ for some $1<r < 2^k$, then we require $2^{k+1}$ $n$-qubit GZZ gates to implement an arbitrary collection of parallel GZZ gates. Note that $2^{k+1} < 2n$ so that regardless of the number of qubits $n$, we can implement the desired targeted gates using less than $2n$ untargeted gates. In particular, we have the following:
\begin{proposition}
	An $n$-qubit CNOT circuit of depth $d$ can be synthesised with local Clifford gates and at most $d2^k$ untargeted $n$-qubit $\GZZ(\frac{\pi}{2^{k+1}})$ gates where $k\in\mathbb{N}$ is the smallest number such that $n-1\leq 2^k$.
\end{proposition}

As is shown in~\cite{jiang2020optimal}, the asymptotically optimal depth of ancilla-free CNOT circuits on $n$ qubits is $O(n/\log(n))$. Our result then implies that we can implement an arbitrary $n$-qubit CNOT circuit using $O(n^2/\log(n))$ $n$-qubit GZZ gates. Interestingly, this is equal to the asymptotically optimal implementation of CNOT circuits using CNOT gates presented in~\cite{markov2008optimal}. 

As any Clifford circuit can be written using local Cliffords and at most three layers of CNOT circuits we get the following result.
\begin{proposition}\label{prop:clifford-untargeted}
	An $n$-qubit Clifford circuit can be implemented using local Clifford gates and $O(n^2/\log(n))$ untargeted $n$-qubit $\GZZ(\frac{\pi}{2^{k+1}})$ gates, where $k$ is the smallest number such that $n-1\leq 2^k$.
\end{proposition}

In Section~\ref{sec:constructingCircuit} we saw how to construct an arbitrary circuit using targeted $\GCZ$ gates. This worked by implementing each non-Clifford phase gate as an exponentiated Pauli gate. While the implementation of such a gate actually requires 2 $\GCZ$ gates, by pushing some leftover gates into the remainder of the circuit we constructed an algorithm using just a single $\GCZ$ gate per non-Clifford phase. We can modify this algorithm to only use untargeted $\GZZ$ gates.

Of the five steps of Algorithm~\ref{alg:general-circuit} we only need to modify the third step: instead of using a $\GCZ_S$ gate, we use a pair of untargeted $\GZZ(\frac\pi4)$ gates with NOT gates on the qubits in $S$ as described in Proposition~\ref{prop:gzz-splits} to get $\GZZ_S(\frac\pi2)\circ \GZZ_{Q\backslash S}(\frac\pi2)$. Note that the first of these gates $\GZZ_S(\frac\pi2)$ is the desired $\GCZ_S$ gate up to local Clifford gates. The other gate $\GZZ_{Q\backslash S}(\frac\pi2)$ is Clifford and affects a disjoint set of qubits, so that it can be pushed onto the remainder of the circuit, as in step 4 of Algorithm~\ref{alg:general-circuit}. 
For example, we can modify the calculation \eqref{eq:phase-gadget-GCZ} as follows:
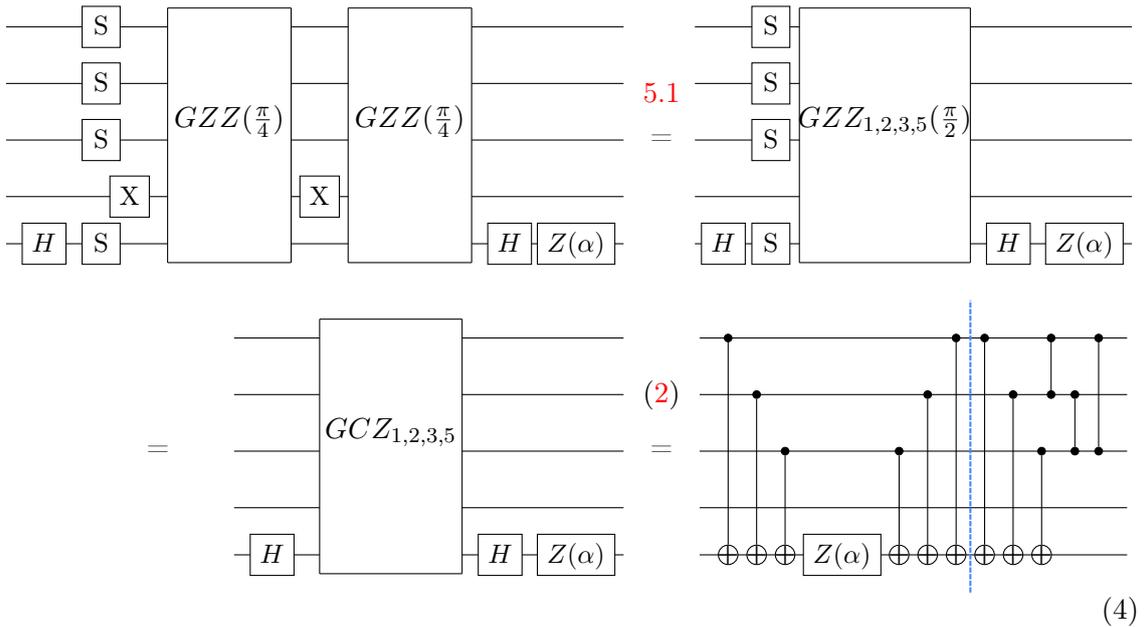
\begin{equation}\label{eq:phase-gadget-GCZ-global}
\begin{tikzpicture}
	\begin{pgfonlayer}{nodelayer}
		\node [style=none] (0) at (1, -1.5) {};
		\node [style=none] (1) at (12.25, -1.5) {};
		\node [style=none] (2) at (1, -3) {};
		\node [style=none] (3) at (12.25, -3) {};
		\node [style=none] (4) at (1, -4.5) {};
		\node [style=none] (5) at (12.25, -4.5) {};
		\node [style=vertex] (9) at (2.5, -3) {};
		\node [style=vertex set] (10) at (2.5, -7.25) {$\!\!+\!\!$};
		\node [style=vertex set] (11) at (3.25, -7.25) {$\!\!+\!\!$};
		\node [style=vertex] (12) at (3.25, -4.5) {};
		\node [style=box] (17) at (4.75, -7.25) {$Z(\alpha)$};
		\node [style=none] (18) at (1, -6) {};
		\node [style=none] (19) at (12.25, -6) {};
		\node [style=vertex] (20) at (7, -3) {};
		\node [style=vertex set] (21) at (7, -7.25) {$\!\!+\!\!$};
		\node [style=vertex set] (22) at (6.25, -7.25) {$\!\!+\!\!$};
		\node [style=vertex] (23) at (6.25, -4.5) {};
		\node [style=vertex] (24) at (1.75, -1.5) {};
		\node [style=vertex set] (25) at (1.75, -7.25) {$\!\!+\!\!$};
		\node [style=vertex] (26) at (7.75, -1.5) {};
		\node [style=vertex set] (27) at (7.75, -7.25) {$\!\!+\!\!$};
		\node [style=none] (28) at (1, -7.25) {};
		\node [style=none] (29) at (12.25, -7.25) {};
		\node [style=none] (31) at (0, 3.75) {=};
		\node [style=none] (32) at (-17.25, 6.75) {};
		\node [style=none] (33) at (-1, 6.75) {};
		\node [style=none] (34) at (-17.25, 5.25) {};
		\node [style=none] (35) at (-1, 5.25) {};
		\node [style=none] (36) at (-17.25, 3.75) {};
		\node [style=none] (37) at (-1, 3.75) {};
		\node [style=none] (43) at (-17.25, 2.25) {};
		\node [style=none] (44) at (-1, 2.25) {};
		\node [style=none] (53) at (-17.25, 1) {};
		\node [style=none] (54) at (-1, 1) {};
		\node [style=none] (55) at (-13, 7.25) {};
		\node [style=none] (56) at (-13, 0.5) {};
		\node [style=none] (57) at (-9.75, 0.5) {};
		\node [style=none] (58) at (-9.75, 7.25) {};
		\node [style=none] (59) at (-13, 6.75) {};
		\node [style=none] (60) at (-13, 5.25) {};
		\node [style=none] (61) at (-13, 3.75) {};
		\node [style=none] (62) at (-13, 2.25) {};
		\node [style=none] (63) at (-13, 1) {};
		\node [style=none] (64) at (-9.75, 1) {};
		\node [style=none] (65) at (-9.75, 2.25) {};
		\node [style=none] (66) at (-9.75, 3.75) {};
		\node [style=none] (67) at (-9.75, 5.25) {};
		\node [style=none] (68) at (-9.75, 6.75) {};
		\node [style=none] (69) at (-11.375, 4.25) {$GZZ(\frac\pi4)$};
		\node [style=box] (70) at (-16.25, 1) {$H$};
		\node [style=box] (71) at (-4, 1) {$H$};
		\node [style=box] (72) at (-2.25, 1) {$Z(\alpha)$};
		\node [style=none] (109) at (8.125, -0.5) {};
		\node [style=none] (110) at (8.125, -8.25) {};
		\node [style=vertex] (111) at (9.25, -3) {};
		\node [style=vertex set] (112) at (9.25, -7.25) {$\!\!+\!\!$};
		\node [style=vertex set] (113) at (10, -7.25) {$\!\!+\!\!$};
		\node [style=vertex] (114) at (10, -4.5) {};
		\node [style=vertex] (115) at (8.5, -1.5) {};
		\node [style=vertex set] (116) at (8.5, -7.25) {$\!\!+\!\!$};
		\node [style=vertex] (117) at (10.875, -4.5) {};
		\node [style=vertex] (118) at (10.875, -3) {};
		\node [style=vertex] (119) at (10.25, -3) {};
		\node [style=vertex] (120) at (10.25, -1.5) {};
		\node [style=vertex] (121) at (11.5, -4.5) {};
		\node [style=vertex] (122) at (11.5, -1.5) {};
		\node [style=none] (152) at (-8.25, 7.25) {};
		\node [style=none] (153) at (-8.25, 0.5) {};
		\node [style=none] (154) at (-5, 0.5) {};
		\node [style=none] (155) at (-5, 7.25) {};
		\node [style=none] (156) at (-8.25, 6.75) {};
		\node [style=none] (157) at (-8.25, 5.25) {};
		\node [style=none] (158) at (-8.25, 3.75) {};
		\node [style=none] (159) at (-8.25, 2.25) {};
		\node [style=none] (160) at (-8.25, 1) {};
		\node [style=none] (161) at (-5, 1) {};
		\node [style=none] (162) at (-5, 2.25) {};
		\node [style=none] (163) at (-5, 3.75) {};
		\node [style=none] (164) at (-5, 5.25) {};
		\node [style=none] (165) at (-5, 6.75) {};
		\node [style=none] (166) at (-6.625, 4.25) {$GZZ(\frac\pi4)$};
		\node [style=box] (167) at (-9, 2.25) {X};
		\node [style=box] (168) at (-14, 2.25) {X};
		\node [style=none] (169) at (0.875, 6.75) {};
		\node [style=none] (170) at (12.375, 6.75) {};
		\node [style=none] (171) at (0.875, 5.25) {};
		\node [style=none] (172) at (12.375, 5.25) {};
		\node [style=none] (173) at (0.875, 3.75) {};
		\node [style=none] (174) at (12.375, 3.75) {};
		\node [style=none] (175) at (0.875, 2.25) {};
		\node [style=none] (176) at (12.375, 2.25) {};
		\node [style=none] (177) at (0.875, 1) {};
		\node [style=none] (178) at (12.375, 1) {};
		\node [style=none] (179) at (3.625, 7.25) {};
		\node [style=none] (180) at (3.625, 0.5) {};
		\node [style=none] (181) at (8.125, 0.5) {};
		\node [style=none] (182) at (8.125, 7.25) {};
		\node [style=none] (183) at (3.625, 6.75) {};
		\node [style=none] (184) at (3.625, 5.25) {};
		\node [style=none] (185) at (3.625, 3.75) {};
		\node [style=none] (186) at (3.625, 2.25) {};
		\node [style=none] (187) at (3.625, 1) {};
		\node [style=none] (188) at (8.125, 1) {};
		\node [style=none] (189) at (8.125, 2.25) {};
		\node [style=none] (190) at (8.125, 3.75) {};
		\node [style=none] (191) at (8.125, 5.25) {};
		\node [style=none] (192) at (8.125, 6.75) {};
		\node [style=none] (193) at (5.875, 4.25) {$GZZ_{1,2,3,5}(\frac\pi2)$};
		\node [style=box] (194) at (1.625, 1) {$H$};
		\node [style=box] (195) at (9.125, 1) {$H$};
		\node [style=box] (196) at (11.125, 1) {$Z(\alpha)$};
		\node [style=box] (197) at (-14.75, 6.75) {S};
		\node [style=box] (198) at (-14.75, 5.25) {S};
		\node [style=box] (199) at (-14.75, 3.75) {S};
		\node [style=box] (200) at (-14.75, 1) {S};
		\node [style=box] (201) at (2.875, 6.75) {S};
		\node [style=box] (202) at (2.875, 5.25) {S};
		\node [style=box] (203) at (2.875, 3.75) {S};
		\node [style=box] (204) at (2.875, 1) {S};
		\node [style=none] (205) at (-11.25, -1.5) {};
		\node [style=none] (206) at (-1, -1.5) {};
		\node [style=none] (207) at (-11.25, -3) {};
		\node [style=none] (208) at (-1, -3) {};
		\node [style=none] (209) at (-11.25, -4.5) {};
		\node [style=none] (210) at (-1, -4.5) {};
		\node [style=none] (211) at (-11.25, -6) {};
		\node [style=none] (212) at (-1, -6) {};
		\node [style=none] (213) at (-11.25, -7.25) {};
		\node [style=none] (214) at (-1, -7.25) {};
		\node [style=none] (215) at (-9, -1) {};
		\node [style=none] (216) at (-9, -7.75) {};
		\node [style=none] (217) at (-5.25, -7.75) {};
		\node [style=none] (218) at (-5.25, -1) {};
		\node [style=none] (219) at (-9, -1.5) {};
		\node [style=none] (220) at (-9, -3) {};
		\node [style=none] (221) at (-9, -4.5) {};
		\node [style=none] (222) at (-9, -6) {};
		\node [style=none] (223) at (-9, -7.25) {};
		\node [style=none] (224) at (-5.25, -7.25) {};
		\node [style=none] (225) at (-5.25, -6) {};
		\node [style=none] (226) at (-5.25, -4.5) {};
		\node [style=none] (227) at (-5.25, -3) {};
		\node [style=none] (228) at (-5.25, -1.5) {};
		\node [style=none] (229) at (-7.125, -4) {$GCZ_{1,2,3,5}$};
		\node [style=box] (230) at (-10.25, -7.25) {$H$};
		\node [style=box] (231) at (-4.25, -7.25) {$H$};
		\node [style=box] (232) at (-2.25, -7.25) {$Z(\alpha)$};
		\node [style=none] (233) at (-13.25, -4.5) {=};
		\node [style=none] (234) at (0, -4.5) {=};
		\node [style=none] (235) at (0, 5) {\ref{prop:gzz-splits}};
		\node [style=none] (236) at (0, -3) {\eqref{eq:phase-gadget-GCZ}};
	\end{pgfonlayer}
	\begin{pgfonlayer}{edgelayer}
		\draw [in=180, out=0] (0.center) to (1.center);
		\draw (3.center) to (2.center);
		\draw [in=180, out=0] (4.center) to (5.center);
		\draw (10) to (9);
		\draw (11) to (12);
		\draw [in=180, out=0] (18.center) to (19.center);
		\draw (21) to (20);
		\draw (22) to (23);
		\draw (25) to (24);
		\draw (27) to (26);
		\draw [in=180, out=0] (28.center) to (29.center);
		\draw (55.center) to (58.center);
		\draw (58.center) to (57.center);
		\draw (57.center) to (56.center);
		\draw (56.center) to (55.center);
		\draw (53.center) to (63.center);
		\draw (62.center) to (43.center);
		\draw (36.center) to (61.center);
		\draw (60.center) to (34.center);
		\draw (32.center) to (59.center);
		\draw [style=hadamard edge] (110.center) to (109.center);
		\draw (112) to (111);
		\draw (113) to (114);
		\draw (116) to (115);
		\draw (117) to (118);
		\draw (120) to (119);
		\draw (122) to (121);
		\draw (152.center) to (155.center);
		\draw (155.center) to (154.center);
		\draw (154.center) to (153.center);
		\draw (153.center) to (152.center);
		\draw (68.center) to (156.center);
		\draw (157.center) to (67.center);
		\draw (66.center) to (158.center);
		\draw (65.center) to (159.center);
		\draw (160.center) to (64.center);
		\draw (161.center) to (54.center);
		\draw (44.center) to (162.center);
		\draw (163.center) to (37.center);
		\draw (164.center) to (35.center);
		\draw (165.center) to (33.center);
		\draw (179.center) to (182.center);
		\draw (182.center) to (181.center);
		\draw (181.center) to (180.center);
		\draw (180.center) to (179.center);
		\draw (192.center) to (170.center);
		\draw (172.center) to (191.center);
		\draw (190.center) to (174.center);
		\draw (176.center) to (189.center);
		\draw (188.center) to (178.center);
		\draw (177.center) to (187.center);
		\draw (186.center) to (175.center);
		\draw (173.center) to (185.center);
		\draw (184.center) to (171.center);
		\draw (169.center) to (183.center);
		\draw (215.center) to (218.center);
		\draw (218.center) to (217.center);
		\draw (217.center) to (216.center);
		\draw (216.center) to (215.center);
		\draw (228.center) to (206.center);
		\draw (208.center) to (227.center);
		\draw (226.center) to (210.center);
		\draw (212.center) to (225.center);
		\draw (224.center) to (214.center);
		\draw (213.center) to (223.center);
		\draw (222.center) to (211.center);
		\draw (209.center) to (221.center);
		\draw (220.center) to (207.center);
		\draw (205.center) to (219.center);
	\end{pgfonlayer}
\end{tikzpicture}
\end{equation}
Note that in this particular case, the $\GZZ_{Q\backslash S}(\frac\pi2)$ acts on just the fourth qubit, so that it is equal to the identity.

We now have the following variation on Theorem~\ref{thm:circuit-with-GMS}.

\begin{theorem}\label{thm:circuit-with-GMS-untargeted}
	An $n$-qubit circuit consisting of Clifford gates and $N$ non-Clifford Z-phase gates can be implemented using single qubit unitaries and at most $2N+O(n^2/\log(n))$ untargeted $\GMS$ gates.
\end{theorem}
\begin{proof}
	Combine the previous discussion with Proposition~\ref{prop:clifford-untargeted}.
\end{proof}

\section{Conclusion}\label{sec:conclusion}

We studied how to transform arbitrary circuits given by single-qubit and two-qubit unitaries into circuits using global multi-qubit gates. 
We found a more efficient construction for Clifford unitaries than the previous best, and we found an algorithm that constructs circuits containing a number of GMS gates proportional to the number of non-Clifford gates in the input circuit, both in the targeted setting where the GMS gate can be applied to any subset of qubits, and the untargeted setting where the GMS gates can only be applied to all qubits at once.
To the author's knowledge these are the first results improving upon the naive algorithm for constructing GMS circuits for \emph{arbitrary} input circuits.

The constructions of~\cite{maslov2018use,ivanov2015efficient,groenland2020signal} that build GMS circuits for specific families of unitaries are better than the general bounds we find here for those specific unitaries. For future work it might therefore be interesting to see how much their results can be adapted to the general case.

Another interesting avenue for future work is to allow GMS gates where the interaction strength varies on each qubit pair. This could model systems where the physical interaction strength lies within an error range of some desired value. The procedure for constructing targeted GMS gates from untargeted ones could potentially be modified so that the error in the interaction strength is suppressed.

\medskip
\noindent \textbf{Acknowledgments}: The author is supported by \censor{a Rubicon fellowship} financed by the \censor{Dutch} \censor{Research Council (NWO)}. The author would like to thank \censor{Aleks Kissinger} for bringing this topic to his attention and for giving valuable suggestions regarding the presentation of the paper.

\bibliographystyle{plainnat}
\bibliography{main}

\appendix

\end{document}